\documentclass[journal]{IEEEtran}
\usepackage{amsmath}
\usepackage{amssymb}
\usepackage{upref}
\usepackage{amsfonts}
\usepackage{graphicx}
\usepackage{psfrag}
\usepackage{verbatim}
\usepackage{pict2e}

\newcommand{\deff}{\mbox{$\stackrel{\rm def}{=}$}}

\newcommand{\field}[1]{\mathbb{#1}}

\newcommand{\Z}{\field{Z}}

\newcommand{\R}{\field{R}}

\newcommand{\cE}{{\cal E}}

\newcommand{\cH}{{\cal H}}
\newcommand{\cA}{{\cal A}}

\newcommand{\cC}{{\cal C}}
\newcommand{\cG}{{\cal G}}
\newcommand{\cL}{{\cal L}}

\newcommand{\cS}{{\cal S}}

\newcommand{\cP}{{\cal P}}

\newcommand{\sP}{\cP}
\newcommand{\sG}{\cG}

\newcommand{\Gr}{\smash{{\sG\kern-1.5pt}_q\kern-0.5pt(n,k)}}
\newcommand{\Gfourk}{\smash{{\sG\kern-1.5pt}_q\kern-0.5pt(4k,2k)}}
\newcommand{\Gk}{\smash{{\sG\kern-1.5pt}_q\kern-0.5pt(n,k_1)}}
\newcommand{\Gkk}{\smash{{\sG\kern-1.5pt}_q\kern-0.5pt(n,k_2)}}
\newcommand{\Grtwo}{\smash{{\sG\kern-1.5pt}_2\kern-0.5pt(n,k)}}
\newcommand{\Gkone}{\smash{{\sG\kern-1.5pt}_q\kern-0.5pt(n,k_1)}}
\newcommand{\Gktwo}{\smash{{\sG\kern-1.5pt}_q\kern-0.5pt(n,k_2)}}
\newcommand{\Ps}{\smash{{\sP\kern-2.0pt}_q\kern-0.5pt(n)}}

\newtheorem{theorem}{Theorem}
\newtheorem{cor}{Corollary}
\newtheorem{lemma}[theorem]{Lemma}

\newtheorem{remark}{Remark}

\newcommand{\ep}{\varepsilon}

\begin{document}

\title{Tilings with $n$-Dimensional Chairs\\ and their Applications to Asymmetric Codes}
\author{Sarit Buzaglo and Tuvi
Etzion,~\IEEEmembership{Fellow,~IEEE}
\thanks{S. Buzaglo is with the Department of Computer Science,
Technion --- Israel Institute of Technology, Haifa 32000, Israel.
(email: sarahb@cs.technion.ac.il). This work is part of her Ph.D.
thesis performed at the Technion.}
\thanks{T. Etzion is with the Department of Computer Science,
Technion --- Israel Institute of Technology, Haifa 32000, Israel.
(email: etzion@cs.technion.ac.il).}
\thanks{This work was supported in part by the Israel
Science Foundation (ISF), Jerusalem, Israel, under Grant 230/08.}
}

\maketitle
\begin{abstract}
An $n$-dimensional chair consists of an $n$-dimensional box
from which a smaller $n$-dimensional box is removed. A tiling
of an $n$-dimensional chair has two nice applications
in some memories using asymmetric codes. The first one is in the design of
codes which correct asymmetric errors with limited-magnitude. The
second one is in the design of $n$~cells $q$-ary write-once memory
codes. We show an equivalence between the design of a tiling
with an integer lattice and the design of a tiling from a generalization
of splitting (or of Sidon sequences).
A tiling of an $n$-dimensional chair can define a perfect code for
correcting asymmetric errors with limited-magnitude.
We present constructions for such tilings and prove cases where
perfect codes for these type of errors do not exist.
\end{abstract}

\begin{keywords}
Asymmetric limited-magnitude errors, lattice, $n\text{-}$dimensional chair,
perfect codes, splitting, tiling, WOM codes
\end{keywords}

\section{Introduction}
\label{sec:introduction}

Storage media which are constrained to change of values
in any location of information only in one direction
were constructed throughout the last fifty years. From the
older punch cards to later optical disks and modern storage
such as flash memories, there was a need to design coding which enables
the values of information to be increased but not to be decreased.
These kind of storage medias are asymmetric memories.
We will call the codes used in these medias, asymmetric codes.
Some of these memories behave as write-once memories
(or WOMs in short)
and coding for them was first considered in the seminal work of Rivest
and Shamir~\cite{RiSh82}. This work initiated a sequence of papers
on this topic, e.g.~\cite{CGM86,FiSh84,FuVi99,WWZK,ZeCo91}.

The emerging new storage media of flash memory raised many new
interesting problems. Flash memory is a nonvolatile reliable memory with
high storage density. Its relatively low cost makes it the ideal
memory to replace the magnetic recording technology in storage media.
A multilevel flash cell is electronically programmed into $q$
threshold levels which can be viewed as elements of the set $\{
0,1, \ldots , q-1 \}$. Raising the charge level of a cell is an easy
operation, but reducing the charge level of a single cell
requires to erase the whole block to which the cell belongs.
This makes the reducing of a charge level to be a complicated, slow, and
unwanted operation. Hence, the cells of the flash memory act as
an asymmetric memory as long as blocks are not erased. This has motivated
new research work on WOMs, e.g.~\cite{CaYa12,Shp12,WQYKS,YaSh12,YSVW12}.

Moreover, usually in programming of the cells,
we let the charge level in a single cell
of a flash memory only to be raised, and hence the
errors in a single cell will be asymmetric.
Asymmetric error-correcting codes were subject to extensive
research due to their applications in coding for computer
memories~\cite{RaFu89}. The errors in a cell of a flash memory are
a new type of asymmetric errors which have limited-magnitude.
Errors in this model are in one
direction and are not likely to exceed a certain limit. This means
that a cell in level $i$ can be raised by an error to level~$j$,
such that $i < j \leq q-1$ and $j-i \leq \ell \leq q-1$,
where $\ell$ is the error limited-magnitude.
Asymmetric error-correcting codes with limited-magnitude were
proposed in~\cite{AAK02} and were first considered for nonvolatile
memories in~\cite{CSBB07,CSBB}. Recently, several other papers
have considered the problem, e.g.~\cite{Dol10,ElBo10,KBE11,YSVW11}.

In this work we will consider a solution for both the construction problem of asymmetric codes
with limited-magnitude and the coding problem
in WOMs. Our proposed solution will use an older
concept in combinatorics named tiling. Given an $n\text{-}$dimensional shape $\cS \subset \Z^n$,
a tiling of $\Z^n$ with $\cS$ consists of disjoint copies of $\cS$
such that each point of $\Z^n$ is covered by exactly one copy of $\cS$.
Tiling is a well established concept in combinatorics and especially
in combinatorial geometry. There are many algebraic methods
related to tiling~\cite{StSz94} and it is an important topic
also in coding theory. For example, perfect codes are
associated with tilings, where the related sphere is the $n$-dimensional
shape $\cS$. Tiling is done with a shape $\cS$ and we consider only shapes which form
an error sphere for asymmetric limited-magnitude codes or their immediate
generalization in $\R^n$. The definition of a tiling in $\R^n$ will be
given in Section~\ref{sec:basic}.

Two of the most considered shapes for tiling
are the cross and the semi-cross~\cite{Ste86,StSz94}. These were also
considered in connections to flash memories~\cite{Sch11}.
In this paper we will consider another shape which
will be called in the \text{sequel} an $n\text{-}$dimensional chair.
An $n$-dimensional chair is an $n\text{-}$dimensional box
from which a smaller $n$-dimensional box is removed
from one of its corners. This is a generalization of the
original concept which is an $n\text{-}$cube from which one vertex
was removed~\cite{LeMo01}. In other places this shape is called
a notched cube~\cite{Kol98,Sch94,Ste90}. Lattice tiling with this shape
will be discussed, regardless of the length of each
side of the larger box and the length of each side of the smaller box.
We will show an equivalent way to present a lattice tiling,
this method will be called a generalized splitting and it generalizes
the concepts of splitting defined in~\cite{Stein67};
and the concept of $B_h[\ell]$ sequences defined and
used for construction of codes correcting asymmetric errors with
limited-magnitude in~\cite{KBE11}.
We will show two applications of tilings with such a shape. One application
is for construction of codes which correct
up to $n-1$ asymmetric limited-magnitude errors
with any given magnitude for each cell; and a second application
is for constructing WOM codes with multiple writing.

In the first part
of this work we will consider only tilings with $n$-dimensional chairs.
In the second part of this work we will consider the
applications of tilings with $n\text{-}$dimensional chairs.
The rest of this work is organized as follows. In
Section~\ref{sec:basic} we define the basic concepts for
our presentation of tilings with $n\text{-}$dimensional chairs.
We define the concepts of an $n\text{-}$dimensional chair and a tiling of the space with
a given shape. We present the $n\text{-}$dimensional chair
as a shape in $\R^n$. When the $n$-dimensional chair consists of unit
cubes connected only by unit cubes of smaller dimensions, the $n\text{-}$dimensional
chair can be represented as a shape in $\Z^n$. For such a shape
we will seek for an integer tiling. We will be interested in this
paper only in lattice tiling and when the shape is in $\Z^n$ only
in integer lattice tiling. Two representations for tiling with a shape
will be given. The first representation is with a generator matrix for
the lattice tiling and the second is by the concept which is called
a generalized splitting.
We will show that these two representations are equivalent.
In Section~\ref{sec:constructSidon} we will present a construction
for tilings with $n$-dimensional chairs based on generalized splitting.
The construction will be based on properties of some Abelian groups.
In Section~\ref{sec:constructLattice} we will present a construction
of tiling with $n$-dimensional chairs based on lattices. This construction
works on any $n$-dimensional chair, while the construction of Section~\ref{sec:constructSidon}
works only on certain discrete ones. We note that after the paper was
written it was brought to our attention that lattice tilings
for notched cubes were given in~\cite{Kol98,Sch94,Ste90}. For completeness
and since our proof is slightly different we kept this part
in the paper. Tiling with a discrete $n$-dimensional chair
can be viewed as a perfect code for correction of asymmetric errors
with limited-magnitude. In Section~\ref{sec:asymmetric} we present the
definition for such codes, not necessarily perfect. We also present
the necessary definition for such perfect codes. We explain what
kind of perfect codes are derived from our constructions and also
how non-perfect codes can be derived from our constructions.
In Section~\ref{sec:nonexist} we prove that certain perfect codes
for correction of asymmetric errors with limited-magnitude do not
exist. In Section~\ref{sec:application}
we will discuss the application of our construction for
multiple writing in $n$ cells $q$-ary write-once memory.
We conclude in Section~\ref{sec:conclude}.

\section{Basic Concepts}
\label{sec:basic}

An $n$-dimensional chair $\cS_{L,K}\subset\R^n$,
$L=(\ell_1,\ell_2,...,\ell_n)$, $K=(k_1,k_2,...,k_n)\in \R^n$,
$0<k_i<\ell_i$ for each $i$, $1\leq i\leq n$, is an $n$-dimensional
$\ell_1 \times \ell_2 \times ~ \cdots ~ \times \ell_n$ box
from which an $n$-dimensional $k_1 \times k_2 \times ~ \cdots ~ \times k_n$
box was removed from one of its corners. Formally, it is defined by
\begin{align*}
\cS_{L,K}=\{ (x_1,x_2,\ldots,x_n )\in \R^n~: ~0 \leq x_i< \ell_i~,~~~~~~~ \\
\hbox{ and there exists a }j,~ 1 \leq j \leq n, \hbox{ such that}~x_j<\ell_j-k_j\}.
\end{align*}

For a given $n$-dimensional shape $\cS$ let $|\cS|$ denote
the \emph{volume} of $\cS$.
The following lemma on the volume of $\cS_{L,K}$ is an immediate
consequence of the definition.
\begin{lemma}
If $L=(\ell_1,\ell_2,...,\ell_n)$, $K=(k_1,k_2,...,k_n)$ are two vectors in
$\R^n$, where $0 < k_i < \ell_i$ for each $i$, $1 \leq i \leq n$,
then
$$|\cS_{L,K}|= \prod_{i=1}^n \ell_i -\prod_{i=1}^n k_i ~.$$
\end{lemma}

\vspace{0.2cm}

If $L=(\ell_1,\ell_2,...,\ell_n),~K=(k_1,k_2,...,k_n)\in \Z^n$
then the $n\text{-}$dimensional chair $\cS_{L,K}$ is
a discrete shape and it can be viewed as a collection of connected
$n$-dimensional unit cubes in which any two adjacent cubes share a complete
{$(n-1)\text{-}$}dimensional unit cube. In this case
the formal definition of the $n$-dimensional
chair, which considers only points of $\Z^n$, is
\begin{align*}
\cS_{L,K}=\{ (x_1,x_2,\ldots,x_n )\in \Z^n~: ~0 \leq x_i< \ell_i~,~~~~~~~ \\
\hbox{ and there exists a }j,~ 1 \leq j \leq n, \hbox{ such that}~x_j<\ell_j-k_j\}.
\end{align*}
\begin{remark}
It is important to note that if ${L=(\ell_1,\ell_2,...,\ell_n),~K=(k_1,k_2,...,k_n)\in \R^n}$
are two integer vectors then the two definitions coincide only
if $\cS_{L,K}$ is viewed as a collection of $n$-dimensional unit cubes.
Special consideration, in the definition, should be given to the boundaries
of the cubes, but this is not an issue for the current work.
\end{remark}

For  $n=2$, if $\ell_1=\ell_2 =\ell$ and $k_1 =k_2 = \ell-1$, then
the chair coincides with the shape known as a corner
(or a semi-cross)~\cite{Stein84}. Examples of a two-dimensional
semi-cross and a three-dimensional chair are given in
Figure~\ref{fig:semi_corner}.

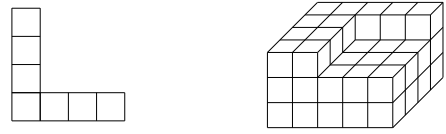
\begin{figure}[tb]

\vspace{2cm}


\centering \setlength{\unitlength}{0.5mm}

\begin{picture}(500,-100)(-10,5)
\resizebox{!}{15mm}{
\linethickness{.4 pt}
\put(50,0){\line(1,0){80}} \put(70,0){\line(0,1){80}}

\put(50,0){\line(0,1){80}} \put(50,20){\line(1,0){80}}

\put(130,0){\line(0,1){20}} \put(50,80){\line(1,0){20}}

\put(110,0){\line(0,1){20}} \put(50,60){\line(1,0){20}}
\put(90,0){\line(0,1){20}} \put(50,40){\line(1,0){20}}
\put(70,0){\line(0,1){20}} \put(50,20){\line(1,0){20}}
}

\end{picture}

\begin{picture}(200,-100)(-80,5)
\resizebox{!}{20mm}{
\linethickness{.4 pt}
\put(70,20){\line(0,1){60}}
\put(50,20){\line(1,0){100}} \put(50,20){\line(0,1){60}}
\put(50,40){\line(1,0){100}} \put(90,20){\line(0,1){60}}
\put(50,60){\line(1,0){100}} \put(110,20){\line(0,1){40}}

\put(50,80){\line(1,0){40}} \put(130,20){\line(0,1){40}}
\put(50,80){\line(1,1){40}} \put(70,80){\line(1,1){40}}
\put(90,80){\line(1,1){40}}
\put(150,20){\line(0,1){40}}
\put(150,20){\line(1,1){40}}
\put(150,40){\line(1,1){40}}
\put(150,60){\line(1,1){40}}
\put(130,60){\line(1,1){30}}
\put(110,60){\line(1,1){30}}
\put(90,120){\line(1,0){100}}
\put(80,110){\line(1,0){100}}
\put(70,100){\line(1,0){40}}
\put(60,90){\line(1,0){40}}
\put(90,60){\line(1,1){30}}
\put(120,90){\line(1,0){60}}
\put(110,80){\line(1,0){60}}
\put(100,70){\line(1,0){60}}
\put(120,90){\line(0,1){20}}
\put(110,80){\line(0,1){20}}
\put(100,70){\line(0,1){20}}
\put(140,90){\line(0,1){20}}
\put(160,90){\line(0,1){20}}
\put(140,110){\line(1,1){10}}

\put(160,110){\line(1,1){10}}
\put(180,110){\line(1,1){10}}
\put(190,60){\line(0,1){60}}
\put(180,50){\line(0,1){60}}
\put(170,40){\line(0,1){40}}
\put(160,30){\line(0,1){40}}

}

\end{picture}

\caption{A semi-cross with $\ell=4$ and a 3-dimensional chair with
         $L=(5,4,3)$ and $K=(3,3,1)$.}

\label{fig:semi_corner}

\end{figure}

A set $P \subseteq \Z^n$ is a \emph{packing} of $\Z^n$ with a
shape $\cS$ if copies of $\cS$ placed on the points of $P$ (in the
same relative position of $\cS$) are disjoint.
A set $T \subseteq \Z^n$ is a \emph{tiling}
of $\Z^n$ with a shape $\cS$ if it is a packing and the disjoint copies of $\cS$
in the packing cover $\Z^n$.

A set $P \subseteq \R^n$ is a \emph{packing} of $\R^n$ with a
shape $\cS$ if copies of $\cS$ placed on the points of $P$ (in the
same relative position of $\cS$) have non-intersecting interiors.
The \emph{closure} of a shape $\cS \subset \R^n$ is
the union of $\cS$ with its exterior surface.
A set $T \subseteq \R^n$ is a \emph{tiling}
of $\R^n$ with a shape $\cS$ if it is a packing and the closure, of
the distinct copies of $\cS$
in the packing, covers $\R^n$.

In the rest of this section we will describe two methods to represent
a packing (tiling) with a shape $\cS$. The first representation is with a lattice. In case that
$\cS$ is a discrete shape we have a
second representation with a splitting sequence.

A \emph{lattice} $\Lambda$ is an additive subgroup of
$\R^n$. We will assume that
\begin{equation*}
\Lambda \deff \{ \lambda_1 V_1 + \lambda_2 V_2 + \cdots + \lambda_n V_n ~:~
\lambda_1, \lambda_2, \cdots, \lambda_n \in \Z \}
\end{equation*}
where $\{ V_1, V_2,\ldots, V_n \}$ is a set of linearly
independent vectors in $\R^n$. The set of vectors $\{ V_1,
V_2,\ldots, V_n\}$ is called \emph{the basis} for $\Lambda$,
and the $n \times n$ matrix
$$
{\bf G} \deff \left[\begin{array}{cccc}
v_{11} & v_{12} & \ldots & v_{1n} \\
v_{21} & v_{22} & \ldots & v_{2n} \\
\vdots & \vdots & \ddots & \vdots\\
v_{n1} & v_{n2} & \ldots & v_{nn} \end{array}\right]
$$
having these vectors as its rows is said to be the \emph{generator
matrix} for $\Lambda$. If $\Lambda \subseteq \Z^n$ then the lattice is called
an \emph{integer lattice}.

The {\it volume} of a lattice $\Lambda$, denoted by $V( \Lambda
)$, is inversely proportional to the number of lattice points per
a unit volume. There is a simple expression for the volume of
$\Lambda$, namely, $V(\Lambda)=| \det {\bf G} |$.

A lattice $\Lambda$ is a \emph{lattice packing
(tiling)} with a shape $\cS$ if the set of points
of $\Lambda$ forms a packing
(tiling) with $\cS$. The following lemma is well known.

\begin{lemma}
\label{lem:vol_lat}
A necessary condition that a lattice $\Lambda$ defines a lattice
packing (tiling) with a shape $\cS$ is that $V(\Lambda) \geq
|\cS|$ ($V(\Lambda)=|\cS|$). A sufficient condition that
a lattice packing $\Lambda$ defines a lattice
tiling with a shape $\cS$ is that $V(\Lambda)=|\cS|$.
\end{lemma}

In the sequel, let $\mathbf{e}_i$ denote the unit vector with an \emph{one}
in the $i$-th coordinate, let $\bf{0}$ denote the all-zero vector,
and let $\bf{1}$ denote the all-one vector.
For two vectors $X=(x_1 , x_2 , \ldots ,
x_n)$, ${Y=(y_1,y_2 , \ldots ,y_n) \in \R^n}$, and a scalar $\alpha \in \Z$,
we define the \emph{vector addition} $X+Y \deff (x_1 + y_1 ,
x_2 + y_2 , \ldots , x_n + y_n )$, and the \emph{scalar
multiplication} $\alpha X
\deff (\alpha x_1 , \alpha x_2 , \ldots , \alpha x_n)$.
For a set $\cS \subset \R^n$ and a vector $U \in \R^n$ the
\emph{shift} of $\cS$ by $U$ is $U + \cS \deff \{ U + X ~:~ X \in
\cS \}$.

Let $G$ be an Abelian group and let
$\beta=\beta_1,\beta_2,...,\beta_n$
be a sequence with $n$ elements of $G$. For every
${X=(x_1,x_2,...x_n)\in \Z^n}$ we define
$$
X\cdot \beta=\sum_{i=1}^n x_i \beta_i,
$$
where addition and multiplication are performed in $G$.

A set $\cS \subset \Z^n$
\emph{splits} an Abelian group $G$ with a \emph{splitting~sequence}
$\beta=\beta_1,\beta_2,...,\beta_n$,
$\beta_i\in G$, for each $i$, $1 \leq i \leq n$, if the set
$\{ \cE\cdot \beta~ ~:~ \cE\in \cS \}$ contains $|\cS|$ distinct elements
from $G$. We will call this operation a \emph{generalized splitting}.
The splitting defined in~\cite{Haj42} and discussed in~\cite{Hic83,Stein67,Stein74,Ste86}
is a special case of the generalized splitting. It was used for
the shapes known as cross and semi-cross~\cite{Stein74,Stein84}, and
quasi-cross~\cite{Sch11}. The $B_h[\ell]$ sequences defined in~\cite{KBE11}
and discussed in~\cite{KBE11,KLNY} for construction of codes which correct
asymmetric errors with limited-magnitude are also a special case of the generalized splitting.
These $B_h [\ell]$ sequences are modification of the well known Sidon sequences
and their generalizations~\cite{Bry04}. The generalized splitting also
makes generalization for a method discussed by Varshamov~\cite{Var64,Var65}.
The generalization can be easily obtained, but to our knowledge a general and complete
proven theory was not given before.

\begin{lemma}
If $\Lambda$ is a lattice packing of $\Z^n$ with a shape
${\cS\subset\Z^n}$ then there exists an Abelian group $G$
of order $V(\Lambda)$, such that $\cS$ splits $G$.
\end{lemma}

\begin{proof}
Let $G=\Z^n/\Lambda$ and let $\phi:\Z^n\rightarrow G$ be the group homomorphism
which maps each element $X\in \Z^n$ to the coset $X+\Lambda$.
Clearly, $|G|=V(\Lambda)$.

Let $\beta= \beta_1 , \beta_2 ,..., \beta_n$,
be a sequence defined by
$\beta_i=\phi(\mathbf{e}_i)$ for each $i$, $1 \leq i \leq n$.
Clearly, for each $X \in \Z^n$ we have ${\phi(X)=X \cdot \beta}$.

Now assume that there exist two distinct elements ${\cE_1,\cE_2\in \cS}$, such that
$$
\phi(\cE_1)=\cE_1\cdot\beta =\cE_2\cdot\beta=\phi(\cE_2)~.
$$
It implies that
$$
\phi(\cE_1-\cE_2)=(\cE_1 - \cE_2) \cdot \beta = \cE_1 \cdot \beta - \cE_2 \cdot \beta =0~.
$$
Since $\phi (X)=0$ if and only if $X \in \Lambda$
it follows that there exists $X\in \Lambda$, $X\neq \mathbf{0}$, such that
$$
\cE_1=\cE_2+X~.
$$
Therefore, $\cS\cap (X+\cS)\neq \varnothing$ which contradicts the fact that $\Lambda$ is a lattice
packing of $\Z^n$ with the shape $\cS$.

Thus, $\cS$ splits $G$ with the splitting sequence $\beta$.
\end{proof}

\begin{lemma}
Let $G$ be an Abelian group and let $\cS$ be a shape in $\Z^n$.
If $\cS$ splits $G$ with a splitting sequence
$\beta$ then there exists a lattice packing $\Lambda$
of $\Z^n$ with the shape $\cS$, for which $V(\Lambda)\leq |G|$.
\end{lemma}

\begin{proof}
Consider the group homomorphism $\phi:\Z^n\rightarrow G$ defined by
$$
\phi(X)=X\cdot \beta.
$$

Clearly, $\Lambda=\ker(\phi)$ is a lattice and the volume
of $\Lambda$, $V(\Lambda)=|\phi(\Z^n)|\leq |G|$.

To complete the proof we have to show that
$\Lambda$ is a packing of $\Z^n$ with the shape $\cS$. Assume to the contrary that
there exists $X\in \Lambda$ such
that $\cS\cap (X+\cS)\neq \varnothing$. Hence, there exist two
distinct elements
$\cE_1,\cE_2\in \cS$ such that $\cE_1=\cE_2+X$ and therefore,
$$
\phi(\cE_1)=\phi(\cE_2+X)=\phi(\cE_2)+\phi(X)=\phi(\cE_2).
$$
Therefore,
$\cE_1\cdot \beta=\cE_2\cdot \beta$, which contradicts the
fact that $\cS$ splits $G$ with the splitting sequence $\beta$.

Thus, $\Lambda$ is a lattice packing with the shape $\cS$.
\end{proof}

\begin{cor}\label{cor:tilingGroup}
A lattice tiling of $\Z^n$ with the shape $\cS\subseteq \Z^n$ exists
if and only if there exists an Abelian group
$G$ of order $|\cS|$ such that $\cS$ splits $G$.
\end{cor}


If our shape $\cS$ is not discrete, i.e. cannot be
represented by a set of $n$-dimensional units cubes,
two of which are adjacent only if they share an $(n-1)$-dimensional
unit cube, then clearly tiling can be represented
with a lattice, but cannot be represented with
a splitting sequence. But, if our shape $\cS$ is in $\Z^n$
then we can use both methods as they were proved
to be equivalent. In fact, both methods are complementary.
If we consider the matrix $\cH =[ \beta_1 ~ \beta_2 ~ \cdots ~ \beta_n]$
then the vector $X=(x_1,x_2,\ldots,x_n) \in \Z^n$
is contained in the related lattice if and only if $\cH X=0$.
Therefore, $\cH$ has some similarity to a parity-check matrix
in coding theory.
The representation of a lattice with its generator matrix
seems to be more practical. But, sometimes it is not
easy to construct one. Moreover, the splitting sequence
has in many cases a nice structure and from its structure
the general structure of the lattice can be found.
This is the case in the next two sections.
In Section~\ref{sec:constructSidon} we present two constructions of tilings
based on generalized splitting. Even though the second one generalizes the
first one, the mathematical structure of the first one has its own beauty
and hence both constructions are given. The construction
of the lattice, in $\R^n$, given in Section~\ref{sec:constructLattice},
was derived based on the structure of the lattices, in $\Z^n$,
obtained from the construction of the splitting sequences
in Section~\ref{sec:constructSidon}.

\section{Constructions based on Generalized Splitting}
\label{sec:constructSidon}

In this section we will present a construction of a tiling with $n$-dimensional
chairs based on generalized splitting. The $n$-dimensional chairs
which are considered in this section are discrete, i.e. $L,K \in \Z^n$.
We start with a construction in which all the
$\ell_i$'s are equal to $\ell$, and all the $k_i$'s
are equal to $\ell-1$. We generalize this construction to a case
in which all the $k_i$'s, with a possible exception of one,
have multiplicative inverses in the related Abelian group.

For the ring $G=\Z_q$, the ring of integers modulo $q$, let $G^*$
be the multiplicative group of $G$ formed from all the elements of
$G$ which have multiplicative inverses in $G$.

\begin{lemma}
\label{lemma:three_property} Let $n \geq 2$, $\ell \geq 2$, be two
integers and let $G$ be the ring of integers modulo
$\ell^n-(\ell-1)^n$, i.e. $\Z_{\ell^n-(\ell-1)^n}$. Then,

\begin{itemize}
\item[(P1)] $\ell-1$ and $\ell$ are elements of $G^*$.

\item[(P2)] $\alpha=\ell (\ell-1)^{-1}$ is an
element of order $n$ in $G^*$.

\item[(P3)] $1+\alpha+\alpha^2+ \cdots +\alpha^{n-1}$ equals to
\emph{zero} in~$G$.
\end{itemize}
\end{lemma}

\begin{proof}
\begin{itemize}
\item[(P1)]
By definition, $\ell^n-(\ell-1)^n$ is \emph{zero} in $G=\Z_{\ell^n-(\ell-1)^n}$.
We also have that ${\ell^n-(\ell-1)^n = \sum_{i=0}^{n-1}{{n}\choose{i}}
(\ell-1)^i}$
${=1+(\ell-1) \sum_{i=1}^{n-1}{{n}\choose{i}} (\ell-1)^{i-1}}$.
It follows that ${(\ell-1) (-\sum_{i=1}^{n-1} {{n}\choose{i}} (\ell-1)^{i-1})=1}$ in $G$,
and hence, $\ell-1 \in G^*$. Since
$\ell^n-(\ell-1)^n$ is \emph{zero} in $G$, it follows that
$\ell^n = (\ell-1)^n$, and hence $\ell \in G^*$
if and only if $\ell-1 \in G^*$.

\item[(P2)] Clearly, $\alpha^n=\ell^n ((\ell-1)^{-1})^n$ and since
$\ell^n=(\ell-1)^n$, it follows that ${\alpha^n=(\ell-1)^n (\ell-1)^{-n}=1}$. This also
implies that $\alpha$ has a multiplicative inverse
and hence $\alpha=\ell (\ell-1)^{-1} \in G^*$.

Now, note that for each $i$, $1\leq i\leq n-1$, we have
${0<\ell^i-(\ell-1)^i<\ell^n-(\ell-1)^n}$.
Therefore, ${\ell^i \neq (\ell-1)^i}$ in $G$
and hence ${\alpha^i=\ell^i ((\ell-1)^{-1})^i\neq 1}$.
Thus, the order of $\alpha$ in $G^*$ is $n$.

\item[(P3)] Clearly, $0=\alpha^n-1=(\alpha-1)(1+\alpha+\alpha^2+...+\alpha^{n-1})$.
By definition, $\alpha=\ell (\ell-1)^{-1}$
and hence $\alpha (\ell -1) = \ell$, $\alpha \ell - \alpha = \ell$,
$\alpha - \alpha \ell^{-1} =1$, $\alpha -1 = \alpha \ell^{-1}$,
$\alpha -1 = (\ell-1)^{-1}$. Therefore,
$0= (\ell-1)^{-1}(1+\alpha+\alpha^2+...+\alpha^{n-1})$ which implies
that $1+\alpha+\alpha^2+...+\alpha^{n-1} =0$.
\end{itemize}
\end{proof}

\begin{theorem}
\label{thm:mainT} Let $n \geq 2$, $\ell \geq 2$, be two integers,
${G=\Z_{\ell^n-(\ell-1)^n}}$, and
$\alpha=\ell (\ell-1)^{-1}$. Then $\cS_{L,K}$,
${L=(\ell,\ell,\ldots,\ell)}$, $K=(\ell-1,\ell-1,\ldots,\ell-1)$,
splits $G$ with
the splitting sequence $\beta=\beta_1,\beta_2,...,\beta_n$ defined by
$$
\beta_i=\alpha^{i-1},~~~~ 1 \leq i \leq n~.
$$
\end{theorem}

\begin{proof}
We will show by induction that every element in $G$ can be
expressed in the form $\cE\cdot \beta$, for some $\cE\in \cS_{L,K}$.

The basis of induction is $0=\mathbf{0}\cdot \beta$.

For the induction step we have to show that
if $x\in G$ can be presented as $x=\cE\cdot \beta$ for some $\cE\in\cS_{L,K}$
(i.e. ${\cE=(\varepsilon_1,\varepsilon_2,...,\varepsilon_n) \in \Z^n}$,
$0\leq \varepsilon_i\leq \ell-1$, $1 \leq i \leq n$,
and for some $j$, $\varepsilon_j=0$),
then also $x+1$ can be presented in the same way. In other words,
$x+1=\tilde{\cE}\cdot \beta$, where
$\tilde{\cE}=(\tilde{\varepsilon}_1,\tilde{\varepsilon}_2,...,\tilde{\varepsilon}_n) \in \cS_{L,K}$.

If $\varepsilon_1<\ell-1$ and there exists $j \neq 1$
such that $\varepsilon_j=0$ then
$$
x+1=\tilde{\cE}\cdot \beta,
$$
where $\tilde{\cE}=\cE+\mathbf{e}_1$, $0\leq \tilde{\varepsilon}_i\leq \ell-1$,
$\tilde{\varepsilon}_j=0$ and the induction step is proved.

If $\varepsilon_1=0$ and there is no $j \neq 1$
such that $\varepsilon_j=0$ then by Lemma~\ref{lemma:three_property}(P3)
we have that $\sum_{i=1}^n \beta_i =0$ and hence
$$
x+1 = (\cE + \mathbf{e}_1 - \mathbf{1}) \cdot \beta ~,
$$
i.e. $\tilde{\cE} = \cE + \mathbf{e}_1 - \mathbf{1}$ is the required
element of $\cS_{L,K}$ and the induction step is proved.

Now, assume that $\varepsilon_1=\ell -1$. Let $j$,
$2\leq j\leq n$ be the smallest index such that
$\varepsilon_j=0$.

$$
x+1=\ell \beta_1 +\sum_{i=2}^n \varepsilon_i \beta_i.
$$
Note that for each $i$, $1\leq i\leq n-1$,
$$
\ell \beta_i=\ell \ell^{i-1} ((\ell-1)^{-1})^{i-1}=(\ell-1) \ell^{i}((\ell-1)^{-1})^{i}=(\ell-1) \beta_{i+1}.
$$
Therefore,
$$
x+1=(\ell-1+\varepsilon_2) \beta_2+\sum_{i=3}^n \varepsilon_i \beta_i.
$$
If $j=2$ then $\tilde{\cE}=(0,\ell-1,\varepsilon_3,\ldots,\varepsilon_n)$
and the induction step is proved.
If $\varepsilon_2 > 0$, i.e. $j > 2$, then
$$
x+1=(\varepsilon_2-1) \beta_2+\ell \beta_2+\sum_{i=3}^n\varepsilon_i \beta_i
$$
$$
=(\varepsilon_2-1) \beta_2+(\ell-1+\varepsilon_3) \beta_3+\sum_{i=4}^n\varepsilon_i \beta_i.
$$
By iteratively continuing in the same manner we obtain
$$
x+1=\sum_{i=2}^{j-1}(\varepsilon_i-1) \beta_i+(\ell-1+\varepsilon_j) \beta_j+\sum_{i=j+1}^n\varepsilon_i \beta_i
$$
and since $\varepsilon_j=0$ we have that
$$\tilde{\cE}=(0,\varepsilon_2-1,\ldots,\varepsilon_{j-1}-1,
\ell-1,\varepsilon_{j+1},\ldots,\varepsilon_n)$$
and the induction step is proved.

Since $|G|=|\cS_{L,K}|$, it follows that the set ${\{ \cE\cdot \beta ~:~ \cE\in \cS_{L,K} \}}$
has $|\cS_{L,K}|$ elements.
\end{proof}

\begin{cor}
For each $n \geq 2$ and $\ell \geq 2$ there exists a lattice
tiling of $\Z^n$ with $\cS_{L,K}$,
$L=(\ell,\ell,\ldots,\ell)$, ${K=(\ell-1,\ell-1,\ldots,\ell-1)}$.
\end{cor}

The next theorem and its proof are generalizations of
Theorem~\ref{thm:mainT} and its proof.

\begin{theorem}\label{thm:generalT}
Let $L=(\ell_1,\ell_2,...,\ell_n)$,
$K=(k_1,k_2,...,k_n)$ be two vectors in $\Z^n$
such that $0< k_i< \ell_i$ for each $i$, $1\leq i\leq n$.
Let $\tau=\prod_{i=1}^n \ell_i$,
$\kappa=\prod_{i=1}^n k_i$,
$G=\Z_{\tau-\kappa}$ and assume that for each $i$, $2\leq i\leq n$,
$k_i\in G^*$. Then $\cS_{L,K}$ splits $G$ with
the splitting sequence $\beta=\beta_1,\beta_2,...,\beta_n$ defined by
$$
\begin{array}{cc}
\beta_1=1~~~~~~~~~~~~~& \\

\beta_{i+1}=k_{i+1}^{-1}\ell_i \beta_i &1\leq i\leq n-1 ~.
\end{array}
$$
\end{theorem}

\begin{proof}
First we will show that $k_1 \beta_1=\ell_n \beta_n$.
Since $\tau-\kappa$ equals zero in $G$, it follows that $\tau=\kappa$ in $G$ and hence
$k_1=\ell_1 \ell_2\cdots \ell_n k_2^{-1} k_3^{-1}\cdots k_n^{-1}$.
Therefore,
$$
\ell_n \beta_n=\ell_n k_n^{-1} \ell_{n-1} \beta_{n-1}=~ \cdots
$$
$$
=\ell_n\ell_{n-1}\cdots \ell_1  k_n^{-1}k_{n-1}^{-1} \cdots k_2^{-1}\beta_1=k_1 \beta_1~.
$$
As an immediate consequence from definition we have that for each $i$, $1 \leq i \leq n-1$,
$$
\ell_i \beta_i = k_{i+1} \beta_{i+1} ~.
$$
Next, we will show that
\begin{equation}
\label{eq:LminusK_B0}
(L-K)\cdot \beta=0.
\end{equation}

$$
(L-K)\cdot \beta=\sum_{i=1}^{n}(\ell_i-k_i) \beta_i=\sum_{i=1}^n (\ell_i \beta_i-k_i \beta_i )
$$
$$
=\ell_n \beta_n-k_n \beta_n+\sum_{i=1}^{n-1} (k_{i+1} \beta_{i+1}-k_i \beta_i )
$$
$$
=\ell_n \beta_n-k_n \beta_n +k_n \beta_n -k_1 \beta_1=0
$$

Since $|\cS_{L,K}|=|G|$ it follows that to prove Theorem \ref{thm:generalT},
it is sufficient to show that each element in $G$ can be
expressed as $\cE\cdot \beta$, for some $\cE\in \cS_{L,K}$.
The proof will be done by induction.

The basis of induction is $0=\mathbf{0}\cdot \beta$.

In the induction step we will show that if $x\in G$
can be presented as $\cE\cdot \beta$ for some $\cE\in\cS_{L,K}$
then the same is true for $x+1$. In other words,
$x+1=\tilde{\cE}\cdot \beta$, where
$\tilde{\cE}=(\tilde{\varepsilon}_1,\tilde{\varepsilon}_2,...,\tilde{\varepsilon}_n) \in \cS_{L,K}$.

Assume
$$
x=\cE\cdot \beta,
$$
where $\cE=(\varepsilon_1,\varepsilon_2,\ldots,\varepsilon_n)$,
$0\leq \varepsilon_i< \ell_i$ for each $i$, and there exists a $j$ such that
$\varepsilon_j<\ell_j-k_j$.

If $\varepsilon_1 < \ell_1 -k_1 -1$ or
if $\varepsilon_1<\ell_1-1$ and there exists $j \neq 1$ such that
$\varepsilon_j < \ell_j -k_j$, then since $\beta_1 =1$ it follows that
$$
x+1=\tilde{\cE}\cdot \beta,
$$
where $\tilde{\cE}=\cE+\mathbf{e_1}$. Clearly,
$0\leq \tilde{\varepsilon}_i \leq \ell_i-1$; $\tilde{\varepsilon}_1< \ell_1 -k_1$
if $\varepsilon_1 < \ell_1 -k_1 -1$ and otherwise
$\tilde{\varepsilon}_j< \ell_j -k_j$. Hence, the induction step is proved.

If $\varepsilon_1 = \ell_1 -k_1 -1$ and there is
no $j \neq 1$ such that $\varepsilon_j < \ell_j -k_j$
then by (\ref{eq:LminusK_B0}) we have that
$(L-K) \cdot \beta =0$ and hence
$$
x+1 = (\cE + \mathbf{e}_1 - (L-K)) \cdot \beta ~,
$$
i.e. $\tilde{\cE} = \cE + \mathbf{e}_1 - L+K$ is the required
element of $\cS_{L,K}$ and the induction step is proved.

Now, assume that $\varepsilon_1=\ell_1-1$. Let $2\leq j\leq n$
be the smallest index such that $\varepsilon_j<\ell_j-k_j$.
$$
x+1=\ell_1 \beta_1+\sum_{i=2}^n \varepsilon_i \beta_i= (k_2+\varepsilon_2) \beta_2+\sum_{i=3}^n \varepsilon_i \beta_i.
$$

If $j=2$ then $\tilde{\cE}=(0,k_2+\varepsilon_2,\varepsilon_3,\ldots,\varepsilon_n)$
and the induction step is proved.
If $\varepsilon_2 \geq \ell_2 -k_2$ then
$$
x+1=\ell_2 \beta_2+(\varepsilon_2-(\ell_2-k_2)) \beta_2+\sum_{i=3}^n \varepsilon_i \beta_i
$$
$$
=(\varepsilon_2-(\ell_2-k_2)) \beta_2+(k_3+\varepsilon_3) \beta_3+\sum_{i=4}^n \varepsilon_i \beta_i.
$$
By iteratively continuing in the same manner we obtain
$$
x+1=\sum_{i=2}^{j-1}(\varepsilon_i-(\ell_i-k_i)) \beta_i+(k_j+\varepsilon_j) \beta_j+\sum_{i=j+1}^n \varepsilon_i \beta_i~,
$$
and since $\varepsilon_j < \ell_j -k_j$ we have that
$$\tilde{\cE}=(0,\varepsilon_2-\ell_2+k_2,\ldots,\varepsilon_{j-1}-\ell_{j-1}+k_{j-1},
k_j+\varepsilon_j,\varepsilon_{j+1},\ldots,\varepsilon_n)$$
is the element of $\cS_{L,K}$, and the induction step is proved.
%
\end{proof}

\begin{cor}\label{cor:intTilingLK}
Let $L=(\ell_1,\ell_2,...,\ell_n)$, $K=(k_1,k_2,...,k_n)$ be two vectors in $\Z^n$
such that $0< k_i< \ell_i$ for each $i$, ${1\leq i\leq n}$.
Let $\tau=\prod_{i=1}^n \ell_i$
and assume that $\gcd(k_i,\tau)=1$
for at least $n-1$ of the $k_i$'s. Then there exists a lattice
tiling of $\Z^n$ with $\cS_{L,K}$.
\end{cor}

\section{Tiling based on a Lattice}
\label{sec:constructLattice}

Next, we consider lattice tiling of $\R^n$ with $\cS_{L,K}\subset\R^n$,
where $L=(\ell_1,\ell_2,...,\ell_n),~K=(k_1,k_2,...,k_n)\in \R^n$.
We want to remind again that Mihalis Kolountzakis and James Schmerl pointed
on~\cite{Kol98,Sch94,Ste90}, where such tiling can be found.
For completeness and since our proof is slightly different we
kept this part in the paper.
For the proof of the next theorem we need the following lemma.

\begin{lemma}
\label{lem:packing_prop}
Let $X=(x_1,x_2,\ldots,x_n) \in \R^n$. Then,
${\cS_{L,K}\cap (X+\cS_{L,K})\neq \varnothing}$ if and
only if $|x_i|<\ell_i$, for $1 \leq i \leq n$,
and there exist integers $j$ and $r$, $1 \leq j,r \leq n$,
such that $x_j<\ell_j-k_j$ and $-(\ell_r-k_r)<x_r$.
\end{lemma}

\begin{proof}
Assume first that $\cS_{L,K}\cap (X+\cS_{L,K})\neq \varnothing$, i.e. there exists
$(a_1,a_2,...,a_n)\in\cS_{L,K}\cap (X+\cS_{L,K})$. By the definition of $\cS_{L,K}$
it follows that

\begin{equation}
\label{eq:int1}
0\leq a_i < \ell_i~,~~\text{for~each}~i,~1 \leq i \leq n~,
\end{equation}
and there exists a $j$ such that
\begin{equation}
\label{eq:int2}
a_j < \ell_j-k_j~.
\end{equation}
Similarly, for $X+\cS_{L,K}$ we have
\begin{equation}
\label{eq:int3}
x_i\leq a_i < x_i+\ell_i~,~~\text{for~each}~i,~1 \leq i \leq n~,
\end{equation}
and there exists an $r$ such that
\begin{equation}
\label{eq:int4}
a_r < x_r+\ell_r-k_r~.
\end{equation}
It follow from (\ref{eq:int1}) and (\ref{eq:int3}) that
$x_i  \leq a_i < \ell_i$ and
${-\ell_i \leq a_i -\ell_i < x_i}$ for each $i$, $1 \leq i \leq n$.
Hence, $|x_i| < \ell_i$ for each $i$, $1 \leq i \leq n$.
It follow from (\ref{eq:int2}) and (\ref{eq:int3}) that $x_j \leq a_j < \ell_j-k_j$.
It follows from (\ref{eq:int4}) and (\ref{eq:int1})
that $x_r > a_r-(\ell_r-k_r) \geq -(\ell_r-k_r)$.

Now, let $X=(x_1,x_2,\ldots,x_n) \in \R^n$ such that $|x_i|<\ell_i$
for each $i$, $1 \leq i \leq n$, and there exist $j,~r$
such that $x_j<\ell_j-k_j$ and $x_r>-(\ell_r-k_r)$.
Consider the point $A=(a_1,a_2,...,a_n) \in \R^n$, where
$a_i=\max\{x_i,0\}$.

By definition, for each $i$, $1\leq i\leq n$,
$$
0\leq a_i< \ell_i
$$
and $a_j<\ell_j-k_j$. Hence, $A\in \cS_{L,K}$.

Clearly, if $x_i<0$ then $a_i=0$
and if $x_i\geq0$ then $a_i=x_i$. In both cases, since $0<x_i+\ell_i$,
it follows that we have
$$
x_i\leq a_i<x_i+\ell_i~.
$$
We also have $0<x_r+\ell_r-k_r$, and therefore ${x_r\leq a_r<x_r+\ell_r-k_r}$.
Hence, $A\in X+\cS_{L,K}$.

Thus,
$A\in \cS_{L,K}\cap (X+\cS_{L,K})$,
i.e. $\cS_{L,K}\cap (X+\cS_{L,K})\neq \varnothing$.
\end{proof}

The next Theorem is a generalization of Corollary \ref{cor:intTilingLK}.

\begin{theorem}\label{thm:tiling}
Let $L=(\ell_1,\ell_2,...,\ell_n) \in \R^n$ and
$K=(k_1,k_2,...,k_n) \in \R^n$, $0<k_i<\ell_i$,
for all $1\leq i\leq n$.
Let $\Lambda$ be the lattice generated by the matrix

$$
{\mathbf{G}} \deff \left[\begin{array}{cccccc}
\ell_1 & -k_2 & 0&0 &\ldots &0 \\
0 & \ell_2 & -k_3 &0 &\ldots & 0 \\
\vdots & \vdots & \ddots & \ddots & \ddots & \vdots\\
0 &  \ldots & 0 &\ell_{n-2} &-k_{n-1} &0 \\
0 & 0 & \ldots & 0 & \ell_{n-1}&-k_{n} \\
-k_1 &0 &\ldots &0 &0 &\ell_n
 \end{array}\right].
$$
Then $\Lambda$ is a lattice tiling of $\R^n$ with $\cS_{L,K}$.
\end{theorem}

\begin{proof}
It is easy to verify that
$V(\Lambda)=|\det \mathbf{G} |= \prod_{i=1}^n \ell_i -\prod_{i=1}^n k_i=|\cS_{L,K}|$.
We will use Lemma~\ref{lem:vol_lat} to show that $\Lambda$ is a tiling of $\R^n$ with $\cS_{L,K}$.
For this, it is sufficient to show that $\Lambda$ is a packing of $\R^n$ with $\cS_{L,K}$.

Let $X\in \Lambda$, $X\neq \mathbf{0}$,
and assume to the contrary that
${\cS_{L,K}\cap (X+ \cS_{L,K}) \neq \varnothing}$.
Since $X \in \Lambda$
it follows that there exist integers
$\lambda_0,\lambda_1,\lambda_2,...,\lambda_n =\lambda_0$,
not all zeros, such that
$x_i=\lambda_i{\ell_i}-\lambda_{i-1}k_i$,
for every $i$, $1\leq i\leq n$. By Lemma~\ref{lem:packing_prop}
we have that for each $i$, $1 \leq i \leq n$,

$$
-\ell_i<\lambda_i\ell_i-\lambda_{i-1}k_i< \ell_i~,
$$
i.e.
$$
\frac{\lambda_{i-1}k_i}{\ell_i}-1<\lambda_i<\frac{\lambda_{i-1}k_i}{\ell_i} +1~.
$$
Since $\lambda_i$ is an integer it follows that
$\lambda_i=\left\lfloor \frac{\lambda_{i-1}k_i}{\ell_i}\right \rfloor$
or $\lambda_i=\left\lceil \frac{\lambda_{i-1}k_i}{\ell_i}\right \rceil$.
For each $i$, $0 \leq i \leq n-1$, if
$\lambda_i= \rho \geq 0$ then since $k_{i+1} < \ell_{i+1}$ we have that
$$
0\leq \left\lfloor \frac{\rho k_{i+1}}{\ell_{i+1}}\right \rfloor \leq
\lambda_{i+1}\leq \left\lceil \frac{\rho k_{i+1}}{\ell_{i+1}}\right \rceil \leq \rho~.
$$
Hence,
\begin{equation}
\label{eq:lambda_run}
0\leq \lambda_{i+1}\leq \lambda_i~.
\end{equation}
Similarly, if $\lambda_i \leq 0$ we have that
$$
\lambda_i\leq \lambda_{i+1}\leq 0~.
$$
If $\lambda_0 \geq 0$ then by~(\ref{eq:lambda_run}) we have
$$
\lambda_0=\lambda_n \leq \lambda_{n-1} \leq \cdots \leq \lambda_1 \leq \lambda_0~,
$$
and hence $\lambda_i=\rho$ for each $i$, $1 \leq i \leq n$.
Similarly, we have $\lambda_i=\rho$ for each $i$, $1 \leq i \leq n$ if $\lambda_0 \leq 0$.
If $\rho>0$ then since~$\rho$ is an integer
we have that $x_i= \rho (\ell_i-k_i)\geq \ell_i-k_i$, for each $i$, $1 \leq i \leq n$. Hence,
there is no $j$ such that $x_j<\ell_j-k_j$, which contradicts Lemma~\ref{lem:packing_prop}.
Similarly, if $\rho <0$ then for each~$i$,
$1 \leq i \leq n$, $x_i= \rho (\ell_i-k_i)\leq -(\ell_i-k_i)$,
and hence there is no $r$
such that $x_r>-(\ell_j-k_j)$, which contradicts Lemma~\ref{lem:packing_prop}.
Therefore, $\rho =0$, i.e. for each
$i$, $0 \leq i \leq n$, $\lambda_i=0$, a contradiction. Hence, $\Lambda$ is a lattice packing
of $\R^n$ with $\cS_{L,K}$

Thus, by Lemma~\ref{lem:vol_lat}, $\Lambda$ is a lattice tiling of $\R^n$
with $\cS_{L,K}$.
\end{proof}

\begin{remark}
Note, that the construction (Theorem~\ref{thm:tiling})
is based on lattices covers all
the parameters of integers which are not
covered in Section~\ref{sec:constructSidon}.
\end{remark}

\section{Asymmetric Errors with Limited-magnitude}
\label{sec:asymmetric}

The first application for a tiling of $\Z^n$ with an $n$-dimensional
chair is in construction of codes of length $n$
which correct asymmetric errors with limited-magnitude.

Let $Q = \{ 0,1, \ldots , q-1 \}$ be an alphabet with $q$ letters.
For a word $X=(x_1,x_2,\ldots , x_n) \in Q^n$, the \emph{Hamming
weight} of $X$, $w_H(X)$, is the number of nonzero entries in $X$,
i.e., $w_H(X) = \left| \{ i~:~ x_i \neq 0 \} \right|$.

A \emph{code} $\cC$ of \emph{length} $n$ over the alphabet $Q$
is a subset of~$Q^n$. A vector $\cE
=(\ep_1,\ep_2,\ldots,\ep_n)$ is a \emph{$t$-asymmetric-error with
limited-magnitude $\ell$} if $w_H (\cE) \leq t$ and $0 \leq \ep_i
\leq \ell$ for each $1 \leq i \leq n$. The sphere $\cS(n,t,\ell)$
is the set of all $t\text{-}$asymmetric-errors with limited-magnitude
$\ell$. A code $\cC \subseteq Q^n$ can correct $t$-asymmetric-errors
with limited-magnitude $\ell$ if for any two codewords
$X_1,~X_2$, and any two $t$-asymmetric-errors with
limited-magnitude $\ell$, $\cE_1,~\cE_2$, such that $X_1 + \cE_1
\in Q^n$, $X_2 + \cE_2 \in Q^n$, we have that $X_1 + \cE_1 \neq X_2 + \cE_2$.

The size of the sphere
$\cS(n,t,\ell)$ is easily computed.
\begin{lemma}
$\left| \cS(n,t,\ell) \right| = \sum_{i=0}^t \binom{n}{i} \ell^i$.
\end{lemma}
\begin{cor}
$\left| \cS(n,n-1,\ell ) \right| = (\ell +1)^n - \ell^n$.
\end{cor}

For simplicity it is more convenient to consider the code~$\cC$ as
a subset of $\Z_q^n$, where all the additions are performed
modulo~$q$. Such a code $\cC$ can be viewed also as a subset of
$\Z^n$ formed by the set $\{ X + qY ~:~ X \in \cC,~ Y \in \Z^n
\}$. This code is an \emph{extension}, from
$\Z_q^n$ to $\Z^n$, of the code $\cC$. Note, in this code there is
a wrap around (of the alphabet) which does not exist if the alphabet
is $Q$, as in the previous code.

A linear code $\cC$, over $\Z_q^n$, which corrects $t$-asymmetric-errors
with limited-magnitude $\ell$, viewed as a subset of $\Z^n$,
is equivalent to an integer lattice packing of $\Z^n$ with the shape
$\cS(n,t,\ell)$. Therefore, we will call this lattice a
\emph{lattice code}.

Let $\cA(n,t,\ell)$ denote the set of lattice codes in $\Z^n$
which correct $t$-asymmetric-errors with limited-magnitude $\ell$.
A code $\cL \in \cA(n,t,\ell)$ is called \emph{perfect} if it
forms a lattice tiling with the shape $\cS(n,t,\ell)$.
By Corollary \ref{cor:tilingGroup} we have
\begin{cor}
\label{cor:per_lattice} A perfect lattice code $\cL\in
\cA(n,t,\ell)$ exists if and only if there exists an Abelian group
$G$ of order $|\cS(n,t,\ell)|$ such that $\cS(n,t,\ell)$ splits $G$.
\end{cor}

A code $\cL\in \cA(n,t,\ell)$ is formed as an extension of a code
over $\Z_q^n$. Assume we want to form
a code $\cC \subseteq \Sigma^n$, where $\Sigma \deff \{
0,1,\ldots , \sigma -1 \}$, which corrects $t$ asymmetric errors
with limited-magnitude $\ell$. Assume that a construction with
a large linear code $\cC \subset \Sigma^n$ does not exist.
One can take a lattice code $\cL\in
\cA(n,t,\ell)$ over an alphabet with $q$ letters $q > \sigma$.
Then a code over the alphabet $\Sigma$ is formed by $\cC \deff \cL  \cap
\Sigma^n$. Note that the code $\cC$ is usually not linear.
This is a simple construction which always works.
Of course, we expect that there will be many alphabets
in which better constructions can be found.

There exists a perfect lattice code $\cL \in
\cA(n,t,\ell)$ for various parameters with
$t=1$~\cite{KBE11,KLNY}. Such codes also
exist for $t=n$ and all $\ell \geq 1$ and for the parameters of
the Golay codes and the binary repetition codes of odd length~\cite{McSl77}.

The existence of perfect codes which correct ($n-1$)-asymmetric-errors
with limited magnitude $\ell$ was proved in~\cite{KLNY}. The related
sphere $\cS(n,n-1,\ell)$ is an $n$-dimensional chair
$\cS_{L,K}$, where $L=(\ell+1,\ell+1,\ldots,\ell+1)$ and
$K=(\ell,\ell,\ldots,\ell)$.
Sections~\ref{sec:constructSidon} and~\ref{sec:constructLattice}
provide constructions for such codes with simpler description
and simpler proofs that these codes are such perfect codes.

In fact, the constructions in these sections provide tilings
of many other related shapes. More than that, there might be scenarios
in which different flash cells can have different limited magnitude.
For example, if for some cells we want to increase the number
of charge levels. In this case we might need a code which correct
asymmetric errors with different limited magnitudes for different cells.
Assume that for the $i$-th cell the limited magnitude is $\ell_i$.
Our lattice tiling with $\cS_{L,K}$,
$L=(\ell_1 +1,\ell_2 +1, \ldots , \ell_n +1) \in \Z^n$,
$K=(\ell_1,\ell_2,\ldots,\ell_n)$, produces the required
perfect code for this scenario.

\section{Nonexistence of some Perfect Codes}
\label{sec:nonexist}

Next, we ask whether perfect codes, which correct asymmetric errors with limited-magnitude,
exist for $t=n-2$. Unfortunately, such
codes cannot exist. The proof for this claim is the goal of this section.
Most of the proof is devoted to the case in which the limited magnitude
$\ell$ is equal to one. We conclude the section with a proof for
$\ell >1$.

For a word $X=(x_1,x_2,\ldots,x_n) \in \Z^n$, we define
$$
N_+(X)= |\{x_i~|~x_i>0\}|,~~~~~~~ N_-(X)=|\{x_i~|~x_i<0\} |.
$$
We say that a codeword $X \in\cL$, $\cL \in \cA(n,t,\ell)$, \emph{covers}
a word $Y \in \Z^n$ if there exists an element $\cE \in \cS(n,t,\ell)$
such that $Y=X+\cE$.

\begin{lemma}
\label{lem:pacSn-2}
Let $\cL \in \cA(n,t,\ell)$, and assume that
there exists $X=(x_1,x_2,\ldots,x_n) \in \cL$, $X\neq \mathbf{0}$,
such that $|x_i|\leq \ell$, for every $i$,
$1\leq i\leq n$. Then, $N_+(X)\geq t+1$ or $N_-(X)\geq t+1$.
\end{lemma}
\begin{proof}
Let $X=(x_1,x_2,\ldots,x_n) \in \cL$, $X\neq \mathbf{0}$,
such that $|x_i|\leq \ell$, for every $i$,
$1\leq i\leq n$. Assume to the contrary that
$N_+(X) \leq t$ and $N_-(X)\leq t$.
Let $\cE^+ = (\varepsilon^+_1,\varepsilon^+_2,\ldots,\varepsilon^+_n)$
where $\varepsilon^+_i=\max\{x_i,0\}$ and
$\cE^- = (\varepsilon^-_1,\varepsilon^-_2,\ldots,\varepsilon^-_n)$ where $\varepsilon^-_i=\max\{-x_i,0\}$.
Clearly, $\cE^+,~\cE^-\in \cS(n,t,\ell)$ and $X+\cE^-=\cE^+$.

Therefore, $\cS(n,t,\ell)\cap (X+\cS(n,t,\ell))\neq \varnothing$,
which contradicts the fact that $\cL \in \cA(n,t,\ell)$. Thus,
$N_+(X)\geq t+1$ or $N_-(X)\geq t+1$.
\end{proof}

\begin{lemma}
\label{lem:cover1}
Let $\cL\in \cA(n,n-2,\ell)$ be a lattice code.
The word $\mathbf{1} \in \Z^n$, the all-one vector,
can be covered only by a codeword of the form
$\mathbf{1}- \lambda \cdot \mathbf{e}_i$, for some $i$, $1\leq i\leq n$;
where $\lambda$ is an integer, $0\leq \lambda \leq \ell$.
\end{lemma}
\begin{proof} Assume that $X \in \cL$ is the
codeword that covers~$\mathbf{1}$.
Then there exists $\cE=(\varepsilon_1,\varepsilon_2,\ldots,\varepsilon_n)\in \cS(n,n-2,\ell)$ such that $X+\cE=\mathbf{1}$,
i.e. $x_i=1-\varepsilon_i$ and therefore, $1-\ell\leq x_i\leq 1$
for each $i$, $1\leq i\leq n$. Since $w_H(\cE)\leq n-2$ it follows that there
are at least two entries which are equal to \emph{one} in $X$.
By Lemma~\ref{lem:pacSn-2}, it follows that $N_{+}(X) \geq n-1$.
Hence, there are at least $n-1$ entries of $X$
which are equal to \emph{one}. Therefore,
$X=\mathbf{1}- \lambda \mathbf{e}_i$ for some $i$, $1\leq i\leq n$;
where $\lambda$ is an integer, $0\leq \lambda \leq \ell$.
\end{proof}

\begin{lemma}\label{lem:cover2}
Let $\cL\in \cA(n,n-2,\ell)$ be a lattice code.
For every $j$, $1\leq j\leq n$, the word
$W_j=\mathbf{1}-\mathbf{e}_j$ can be covered only by a codeword of the form
$\mathbf{1}- \lambda \mathbf{e}_j$,
where $\lambda$ is an integer, $1\leq \lambda \leq \ell+1$.
\end{lemma}
\begin{proof}
Assume that $X\in \cL$ is a codeword that covers $W_j$.
Then there exists
$\cE=(\varepsilon_1,\varepsilon_2,\ldots,\varepsilon_n)\in \cS(n,n-2,\ell)$ such that
$X+\cE=W_j$. Clearly, $x_j=-\varepsilon_j\leq 0$, and for each $i\neq j$,
$x_i=1-\varepsilon_i$ and therefore $-\ell\leq x_i\leq 1$ for each $i$, $1 \leq i \leq n$.
Since $w_H(\cE)\leq n-2$ it follows that there are at most $n-2$
negative coordinates in~$X$. Therefore,
by Lemma~\ref{lem:pacSn-2}, it follows that $N_{+}(X) \geq n-1$.
Hence, there are at least $n-1$ coordinates of $X$ which
are equal to \emph{one}. Thus, $X=\mathbf{1}-\lambda \mathbf{e}_j$,
where $1\leq \lambda \leq \ell+1$.
\end{proof}

\begin{lemma}
\label{lem:perfect_cond}
If there exists a perfect lattice code in
${\cA(n,n-2,\ell)}$ then $|\cS(n,n-2,\ell)|$ divides
${(\ell+1)^{n-2} (\ell+1+\lambda (n-2-\ell))}$ for some
integer $\lambda$, $0\leq \lambda \leq \ell$.
\end{lemma}
\begin{proof}
Let $\cL\in \cA(n,n-2,\ell)$ be a perfect lattice code.
By Lemma \ref{lem:cover1} and w.l.o.g we can assume that $\mathbf{1}$
is covered by the codeword $X=\mathbf{1}- \lambda \mathbf{e}_n$,
where $0\leq \lambda \leq \ell$. Combining this with
Lemma \ref{lem:cover2} we deduce that for all $i$, $1\leq i\leq n-1$,
the word $W_i=\mathbf{1}-\mathbf{e}_i$ is covered by the codeword
$Y_i=\mathbf{1}-(\ell+1)\cdot \mathbf{e}_i$
($Y_i$ cannot be equal $\mathbf{1}- \alpha \mathbf{e}_i$,
$1 \leq \alpha \leq \ell$ since it would cover $\mathbf{1}$
which is already covered by $X$). We have $n$ distinct
codewords in $\cL$, and since $\cL$ is a lattice, the lattice $\cL'$
generated by the set $\{X,Y_1,Y_2,\ldots,Y_{n-1}\}$ is a sublattice of $\cL$, and
therefore $V(\cL)=|\cS(n,n-2,\ell)|$ divides $V(\cL')$.
Let $\mathbf{G}$ be the matrix whose rows are $X,Y_1,Y_2,\ldots,Y_{n-1}$.

$$
\det\mathbf{G}=\left|\begin{array}{cccccc}
1 & 1 & 1& \ldots & 1& 1-\lambda \\
-\ell & 1 &  1& \ldots  &1& 1 \\
1 & -\ell & 1& \ldots & 1&1\\
1 & 1& -\ell & \ddots &1 & 1 \\
\vdots &\vdots & \ddots& \ddots & \vdots & \vdots \\
1 &1 & 1 & \ldots  &-\ell &1
\end{array}\right|
$$
Subtracting the first row from every other row, we obtain the
determinant
$$
\left|\begin{array}{cccccc}
1 & 1 & 1& \ldots &1& 1-\lambda \\
-(\ell+1) & 0 &  0&\ldots &0& \lambda\\
0 & -(\ell+1) & 0&\ldots &0 & \lambda\\
0 &0 & -(\ell+1)& \ddots & 0& \lambda \\
\vdots & \vdots&\ddots &\ddots&\vdots & \vdots \\
0 &0 &0 &\ldots &-(\ell+1) & \lambda
 \end{array}\right|.
$$
Subtracting the first column from all the other columns, except from the
last one, we obtain the determinant
$$
\left|\begin{array}{cccccc}
1 & 0 & 0& \ldots &0& 1-\lambda \\
-(\ell+1) & \ell+1 &  \ell+1&\ldots &\ell+1& \lambda\\
0 & -(\ell+1) & 0&\ldots &0 & \lambda\\
0 &0 & -(\ell+1)& \ddots & 0& \lambda \\
\vdots & \vdots&\ddots &\ddots&\vdots & \vdots \\
0 &0 &0 &\ldots &-(\ell+1) & \lambda
\end{array}\right|.
$$
Finally, replacing the second row by the sum of all the rows,
except for the first one, we obtain the determinant
$$
\left|\begin{array}{cccccc}
1 & 0 & 0& \ldots &0& 1-\lambda \\
-(\ell+1) & 0 &  0&\ldots &0 & \lambda (n-1)\\
0 & -(\ell+1) & 0&\ldots &0 & \lambda\\
0 &0 & -(\ell+1)& \ddots & 0& \lambda \\
\vdots & \vdots&\ddots &\ddots&\vdots & \vdots \\
0 &0 &0 &\ldots &-(\ell+1) & \lambda
 \end{array}\right|.
$$
Now, it is easy to verify that
$V(\cL')= |\det(\mathbf{G})|=|{\lambda (n-1)} {(\ell+1)^{n-2}}+(1-\lambda ){(\ell+1)^{n-1}}|
=|{(\ell+1)^{n-2}} {(\ell+1+ \lambda (n-2-\ell))}|$.
\end{proof}
\begin{theorem}
\label{thm:NoPerelleq1}
There are no perfect lattice codes in ${\cA(n,n-2,1)}$ for all
$n \geq 4$.
\end{theorem}
\begin{proof}
By Lemma~\ref{lem:perfect_cond}, it is sufficient
to show that ${|\cS(n,n-2,1)|=2^n-n-1}$ does not divide
${2^{n-2} (2+ \lambda (n-3))}$, for $\lambda=0,1$.

If $\lambda =0$ then we have to show that
$2^n-n-1$ does not divide $2^{n-1}$.
It can be readily verified that
$2^n-n-1> 2^{n-1}$ for all $n>3$, which proves
the claim.

If $\lambda =1$ then we have to show that $2^n-n-1$
does not divide $2^{n-2} (n-1)$.
If $2^r = \gcd(2^n-n-1,2^{n-2})$ then ${0\leq r\leq \log_2(n+1)}$.
Hence, we have to show that
$2^{n-r}-\frac{n+1}{2^r}$ does not divide $n-1$.
We will show that for all $n\geq 7$,
$2^{n-r}-\frac{n+1}{2^r}>n-1$.
It is easy to verify that
$$
2^{n-r}-\frac{n+1}{2^r}\geq 2^{n-\log_2(n+1)}-(n+1)=\frac{2^n}{n+1}-n-1~.
$$
Therefore, it is sufficient to show that
$$
\frac{2^n}{n+1}-n-1>n-1~,
$$
or equivalently
$$
2^n>2n (n+1).
$$
This is simply proved by induction on $n$ for all $n\geq 7$.

To complete the proof we should only verify that
for $n=4$, 5, and 6, we have that
$2^n-n-1$ does not divide $2^{n-2} (n-1)$.
\end{proof}

\begin{theorem}
\label{thm:NoPerellgeq2}
There are no perfect lattice codes in ${\cA(n,n-2,\ell)}$ if
$n \geq 4$ and $\ell \geq 2$.
\end{theorem}
\begin{proof}
Let $n \geq 4$ and $\ell \geq 2$ and assume to the contrary,
that there exists a perfect lattice code $\cL\in \cA(n,n-2,\ell)$.
Without loss of generality, we can assume by
Lemma~\ref{lem:cover1} that the word $\mathbf{1} \in \Z^n$ is covered
by a codeword $X=\mathbf{1}- \lambda \mathbf{e}_n$,
where $\lambda$ is an integer, $0\leq \lambda \leq \ell$.
From the proof of Lemma~\ref{lem:perfect_cond} we have
that for all $i$, $1\leq i\leq n-1$,
the word $W_i=\mathbf{1}-\mathbf{e}_i$ is covered by the codeword $X_i=\mathbf{1}-(\ell+1)\cdot \mathbf{e}_i$.
Therefore, $Y=(y_1,y_2,\ldots,y_n)=X_1+X_2=2\cdot\mathbf{1}-
(\ell+1)\cdot \mathbf{e}_1-(\ell+1)\cdot\mathbf{e}_2$ is a codeword
Clearly, $y_1=y_2=2-(\ell+1)=1-\ell$ and since $\ell\geq 2$
it follows that for all $i$, $1\leq i\leq n$, $|y_i|\leq \ell$.
Moreover, $N_-(X)=2\leq n-2$ and $N_+(X)=n-2$, which
contradicts Lemma \ref{lem:pacSn-2}.
Thus, if $n \geq 4$ and $\ell \geq 2$, then
there are no perfect lattice codes in $\cA(n,n-2,\ell)$.
\end{proof}

Combining Theorems~\ref{thm:NoPerelleq1}
and~\ref{thm:NoPerellgeq2} we obtain the main result
of this section.
\begin{cor}
There are no perfect lattice codes in ${\cA(n,n-2,\ell)}$ if
$n \geq 4$ for any limited magnitude $\ell \geq 1$.
\end{cor}

The existence of perfect lattice codes in $\cA(n,n-1,\ell)$
and their nonexistence in $\cA(n,n-2,\ell)$ might give an evidence
that such perfect codes won't exists in $\cA(n,n-\epsilon,\ell)$
for $\ell \geq 1$ and some $\epsilon > 1$. It would be interesting
to prove such a claim for $n \geq 4$ and
$2 \leq \epsilon \leq \lfloor \frac{n}{2} \rfloor$.

\section{Application to Write-Once Memories}
\label{sec:application}

A second possible application for a tiling of $\Z^n$ with an $n$-dimensional
chair is in constructions of multiple writing
in $n$ cells write-once memories. Each cell has $q$ charge
levels $\{ 0,1, \ldots , q -1 \}$. A letter from
an alphabet of size $\sigma$, $\Sigma = \{ 0,1\ldots,\sigma -1 \}$,
is written into the $n$ cells
as many times as possible. In each round the charge level
in each cell is greater than or equal to the charge level in the previous
round. It is desired that the number of rounds
for which we can guarantee to write a new symbol from $\Sigma$ will
be maximized.

An optimal solution for the problem can be described as follows.
Let $A$ be an $q \times q \times \cdots \times q$ $n$-dimensional array.
Let $\psi : A \rightarrow \Sigma$ be a coloring of the array $A$ with
the $\sigma$ alphabet letters. The rounds of writing and raising
the charge levels of the $n$ cells can be described in terms of
the coloring $\psi$ of the array $A$. If in the first round the symbol $s_1$
is written and the charge level in
cell $i$ is raised to $c_i^1$, $1 \leq i \leq n$, then
the color in position $(c_1^1,c_2^1,\ldots,c_n^1)$ is $s_1$. Therefore,
we have to find a coloring function $\psi$ such that the number of rounds
in which a new symbol can be written will be maximal.

Cassuto and Yaakobi~\cite{CaYa12} have found that using a coloring~$\psi$
based on a lattice tiling $\Lambda$ with a two-dimensional
chair provides the best known writing strategy
when there are two cells. A coloring $\tilde{\psi}$
of $\Z^n$ based on a lattice tiling $\Lambda$
with a shape~$\cS$ has $|\cS|$ colors. The lattice have $|\cS|$ cosets,
and hence~$|\cS|$ coset representatives, $X_0,X_1,\ldots,X_{|\cS|-1}$.
The points in $\Z^n$ of the coset $X_i + \Lambda$ are colored with the
$i$-th letter of $\Sigma$. Now, the coloring of
entry $(x_1,x_2,\ldots,x_n)$ of $A$ given by~$\psi$
is equal to the color of the point
$(x_1,x_2,\ldots,x_n) \in \Z^n$ given by the coloring $\tilde{\psi}$.
The method given in~\cite{CaYa12} suggests that a generalization
using coloring based on tiling of $\Z^n$ with an $n\text{-}$dimensional chair
will be a good strategy for WOM codes with $n$ cells~\cite{Yaak12}.
The analysis with two cells, i.e. two-dimensional tiling was discussed
with more details in~\cite{CaYa12}. The analysis for the
$n$-dimensional case will
be discussed in research work which follows by the
same authors and another group as well~\cite{Yaak12}.

\section{Conclusion}
\label{sec:conclude}

We have presented a few constructions for tilings
with $n\text{-}$dimensional chairs. The tilings are based either
on lattices or on generalized splitting. Both methods
are equivalent if our space is $\Z^n$. The generalized splitting
is a simple generalization for known concepts such as splitting
and $B_h[\ell]$ sequences. We have shown that
our tilings can be applied in the design of codes which
correct asymmetric errors with limited-magnitude.
We further mentioned a possible application
in the design of WOM codes for multiple writing.
Finally, we proved that some perfect codes for
correction of asymmetric errors with limited-magnitude
cannot exist.

\vspace{0.5cm}

\begin{center}
{\bf Acknowledgement}
\end{center}

The authors would like to thank Mihalis Kolountzakis and James Schmerl for
pointing on references~\cite{Kol98,Sch94,Ste90},
which consider lattice tiling of notched cubes.

%
%

\newpage
{\bf Sarit Buzaglo} was born in Israel in 1983. She received
the B.A. and M.Sc. degrees from the Technion - Israel Institute of
Technology, Haifa, Israel, in 2007 and 2010, respectively, from the
department of Mathematics.
She is currently a Ph.D. student in the
Computer Science Department at the Technion. Her research
interests include algebraic error-correction coding, coding
theory, discrete geometry, and combinatorics.

\vspace{1cm}

{\bf Tuvi Etzion} (M'89-SM'94-F'04) was born in Tel Aviv, Israel,
in 1956. He received the B.A., M.Sc., and D.Sc. degrees from the
Technion - Israel Institute of Technology, Haifa, Israel, in 1980,
1982, and 1984, respectively.

From 1984 he held a position in the department of Computer Science
at the Technion, where he has a Professor position. During the
years 1986-1987 he was Visiting Research Professor with the
Department of Electrical Engineering - Systems at the University
of Southern California, Los Angeles. During the summers of 1990
and 1991 he was visiting Bellcore in Morristown, New Jersey.
During the years 1994-1996 he was a Visiting Research Fellow in
the Computer Science Department at Royal Holloway College, Egham,
England. He also had several visits to the Coordinated Science
Laboratory at University of Illinois in Urbana-Champaign during
the years 1995-1998, two visits to HP Bristol during the summers
of 1996, 2000,  a few visits to the department of Electrical
Engineering, University of California at San Diego during the
years 2000-2012, and several visits to the Mathematics department
at Royal Holloway College, Egham, England, during the years
2007-2009.

His research interests include applications of discrete
mathematics to problems in computer science and information
theory, coding theory, and combinatorial designs.

Dr Etzion was an Associate Editor for Coding Theory for the IEEE
Transactions on Information Theory from 2006 till 2009.

\end{document}